\documentclass[submission,copyright,creativecommons]{eptcs}
\providecommand{\event}{GANDALF 2021} 

\usepackage{amsmath}
\usepackage{amssymb}
\usepackage{stmaryrd}
\usepackage{amsthm}
\usepackage{booktabs}
\usepackage{mathtools}
\usepackage{tikz}
\usepackage[linguistics]{forest}
\usetikzlibrary{positioning,fit,arrows,backgrounds}
\usetikzlibrary{automata}
\usepackage{xcolor}

\newtheorem{problemstatement}{Problem Statement}

\newtheorem{theorem}{Theorem}

\newtheorem{definition}{Definition}

\newtheorem{assume}{Assumption}

\newcommand{\set}[3]{{#1}_{#2}^{#3}}

\newcommand{\gtuple}{\langle V,\Sigma,\Gamma,E,\iota,F \rangle}

\newcommand{\A}{\mathsf{A}}
\newcommand{\lang}{\mathcal{L}}

\newcommand{\player}[1]{$P_{#1}$}

\title{Adapting to the Behavior of Environments with Bounded Memory\thanks{This material is based upon work supported by ARL ACC-APG-RTP W911NF1920333, AFRL FA9550-19-1-0169, DARPA D19AP00004 and NSF 1652113.}} 
\author{Dhananjay Raju
\institute{The University of Texas at Austin, USA}
\email{draju@cs.utexas.edu} 
\and 
R\"{u}diger Ehlers\institute{Clausthal University of Technology, Germany} \email{ruediger.ehlers@tu-clausthal.de}
\and 
Ufuk Topcu \institute{The University of Texas at Austin, USA} \email{utopcu@utexas.edu}}
\def\titlerunning{Adapting to the Behavior of Environments with Bounded Memory}
\def\authorrunning{Dhananjay Raju, R\"{u}diger Ehlers \& Ufuk Topcu}
\begin{document}
\maketitle
\begin{abstract}
We study the problem of synthesizing implementations from temporal logic specifications that need to work correctly in all environments that can be represented as transducers with a limited number of states.
This problem was originally  defined and studied by Kupferman, Lustig, Vardi, and Yannakakis. They provide NP and 2-EXPTIME lower and upper bounds (respectively) for the complexity of this problem, in the size of the transducer. We tighten the gap by providing a PSPACE lower bound, thereby showing that algorithms for solving this problem are unlikely to scale to large environment sizes.
This result is somewhat unfortunate as solving this problem enables tackling some high-level control problems in which an agent has to infer the environment behavior from observations. To address this observation, we study a modified synthesis problem in which the synthesized controller must gather information about the environment's behavior \emph{safely}. 
We show that the problem of determining whether the behavior of such an environment can be safely learned is only co-NP-complete. Furthermore, in such scenarios, the behavior of the environment can be learned using a Turing machine that requires at most polynomial space in the size of the environment's transducer.
\end{abstract}

\section{Introduction}

Reactive synthesis is the process of automatically computing correct (by construction) implementations of systems from their formal specifications~\cite{Bloem2018,EhlersKB15,DBLP:conf/tacas/MajumdarPS19}. A synthesized system is guaranteed to satisfy its specification along all of its executions, regardless of how the environment behaves. In this way, synthesis is much stronger than \emph{planning} (in deterministic domains), i.e., the process of finding one execution of a system satisfying the specification from its current state \cite{DBLP:journals/jair/PistoreV07}, as synthesis includes planning for all possible behaviors of the environment.
However, there are many cases in which reactive synthesis fails because there is no system that satisfies the specification against \emph{all} environment behaviors. This is for instance the case in scenarios in which the environment can block the system from achieving its objectives~\cite{Kress-Gazit2018May}. 
An example for such a case is depicted in Figure~\ref{fig:robotScenario}, where a human and a robot share a workspace.

\begin{figure}[h]
    \centering
    \begin{tikzpicture}[thick,auto,>=stealth,shorten >=3pt,scale = 0.7] 
    
    \path (-2,0) edge (0,0);
    \path (-2,0) edge (0,2);
    \path (-2,0) edge (0,-2);
    
    \path (2,0) edge (0,0);
    \path (2,2) edge (0,2);
    \path (2,-2) edge (0,-2);
    
    \path (2,0) edge (4,0);
    \path (2,2) edge (4,2);
    \path (2,-2) edge (4,-2);

    \path (6,0) edge (4,0);
    \path (6,0) edge (4,2);
    \path (6,0) edge (4,-2);
    
    \node at (2,3) {charging station 1};
    \node at (2,-.8) {charging station 2};
    \node at (2,-3.2) {charging station 3};
    
    \node at (0,2.5) {lane~1};
    \node at (0,.5) {lane~2};
    \node at (0,-2.5) {lane~3};
    
    \node at (4,2.5) {lane~1};
    \node at (4,.5) {lane~2};
    \node at (4,-2.5) {lane~3};
    
    \node at (-5,0) {initial robot location};
    \node at (9,0) {initial human location};
    
    \node at (-2,0) {\includegraphics[width=0.8cm,height=0.8cm]{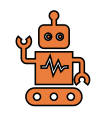}};
    \node at (2,0) {\includegraphics[width=0.8cm,height=0.8cm]{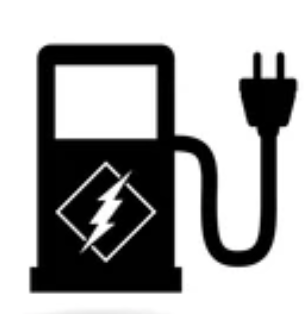}};
    \node at (2,-2) {\includegraphics[width=0.8cm,height=0.8cm]{battery.png}};
    \node at (2,2) {\includegraphics[width=0.8cm,height=0.8cm]{battery.png}};
    \node at (6,0) {\includegraphics[width=0.8cm,height=0.8cm]{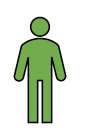}};
    \draw[fill = black] (0,0) circle (3pt);
    \draw[fill = black] (0,2) circle (3pt);
    \draw[fill = black] (0,-2) circle (3pt);
    \draw[fill = black] (4,0) circle (3pt);
    \draw[fill = black] (4,2) circle (3pt);
    \draw[fill = black] (4,-2) circle (3pt);
    
    \end{tikzpicture}
    \caption{A robot and a human are operating in the same environment. They take turns to move to an adjacent location. 
    The objective of the robotic agent is to use a charging station without colliding with the human. The objective of the human is unknown to the robot.}
    \label{fig:robotScenario}
\end{figure}


The robot controller to be synthesized has the task of evading the human while at the same time recharging whenever necessary by visiting a recharging station and staying there for two time steps in a row. After formalizing the scenario in the form of a specification, a reactive synthesis tool will conclude that there is no implementation, as the human can always move to the same lane as the robot. This human behavior requires the robot to back up as it is its task to evade. In this way, the robot can never use a recharging station.
While the answer that there is no controller is correct, a human engineer would typically write a controller by hand for this scenario that manages to recharge correctly whenever the human at least occasionally stays clear from the robot, for instance to pursue its own goals in the workspace. 

To weaken the overly strong requirement that the synthesized controller has to always operate correctly no matter how the environment behaves, a common approach is to make specific \emph{assumptions} that restrict the environment's behavior. The synthesized system then has to operate correctly only in environments that satisfy these assumptions. However, this approach creates a new issue that the synthesized implementations are incentivized to work against the satisfaction of the assumptions (which can be partially mitigated \cite{DBLP:conf/atva/BloemEK15,DBLP:conf/tacas/MajumdarPS19}). Alternatively, the interaction between environment and system can be viewed from the perspective of strategic games, in which some form of stable equilibrium between the strategies of the environment and system players is computed such that none of the two players are incentivized to deviate \cite{DBLP:journals/acta/HunterPR17,DBLP:journals/acta/BrenguierRS17}. Both approaches require information about the environment's goals to perform this strategic reasoning.
Such information about the environment in which the system to be synthesized is supposed to operate in is unfortunately not always available. 

In the example from Figure~\ref{fig:robotScenario}, we are seeking a controller that avoids collisions with a human without knowing the human's intention in the shared workspace. A controller can do so by \emph{observing} the behavior of the human and adapting its control policy in a way that the human is avoided.
Since environments can behave arbitrarily within their limits (in synthesis), they can also change their behavior arbitrarily and hence past observed behavior is useless in the setting of classical reactive synthesis.
This observation leads to the question if an alternative definition for the synthesis problem exists that would enable us to perform formal synthesis of a correct-by-construction controller in unknown environments. Some quantification about the environment's capabilities is necessary to make solving this problem useful, as otherwise the environment can behave fully antagonistically as in classical reactive synthesis. At the same time, we are seeking for a controller that always works correctly against simple behaviors of the environment.  
This requirement can be formalized by 
starting from the observation that the size of the memory of a \emph{transducer} encoding the environment's behavior can be used as an abstract notion of the environment's behavioral complexity \cite{Papadimitriou2020}. We note that the need to bound the environment is of interest in several other paradigms in computer science. For example, in cryptography, one studies the security of a given crypto-system with respect to
attackers with bounded computational power \cite{Borodin1998May}. A symbolic synthesis procedure for bounded synthesis of lasso-precise implementations based on quantified Boolean formula solving has been provided in \cite{ryana_lassos}.

In the  scenario from Figure~\ref{fig:robotScenario}, simple behavior of the human (environment) would enable the robot to perform its task, while very complex behavior such as blocking the robot by always moving to the same lane as the robot does not.
The human needs to use a state of memory for moving to a particular charging lane. Additionally, it needs one memory state to come back to the initial location.  
Say the human is using a transducer with three states, then the human can block the robot from charging in at most two lanes. However, with four states, the human can successfully block the robot from charging.
The \emph{synthesis under bounded environments} problem has originally been defined by  Kupferman, Lustig, Vardi and Yannakakis \cite{boundedtemp}, who also provide an algorithm for this synthesis problem that has a time complexity that is doubly-exponential in the number of states of finite-state machines (also known as \emph{transducers}) representing the environment. However, the lower bound on the complexity that they give is only NP, leaving a hope that this problem can be solved for scenarios of practical relevance, e.g., by employing a satisfiability (SAT) solver.


In this paper, we tighten the gap between the upper and lower complexity bounds of synthesis under bounded environments and provide a new PSPACE lower bound, shattering the hope that the synthesis problem for bounded environments has a (relatively) low complexity. Our proof is based on the observation that in order to solve the problem, the synthesis algorithm needs to distill the safe ways of ``probing'' the environment in order to obtain information about its behavior, which causes the high complexity. 
We then prove that if the synthesized controller can observe the environment in a \emph{safe way}, the problem is simpler. This applies, for instance, to the robotics scenario from Figure~\ref{fig:robotScenario}. 
The robot can always move safely without restricting future behavior due to its past actions. However, for eventually recharging, the robot needs to use the observations made.

To precisely capture such scenarios, we strengthen the original problem formulation by Kupferman, Lustig, Vardi and Yannakakis and  
introduce the notion of $k$-\emph{transducer liveness}. A synthesis problem instance is $k$-transducer live if from every prefix behavior that is compatible with at least one transducer of size $k$ (for the environment), there is a way for the system to continue operating such that it can eventually satisfy its objectives. Computing a controller for the system requires that no ``probing'' can make the system \emph{irrecoverable}, i.e., get the system into a situation from which it cannot satisfy its specification due to its prior actions.
The example in Figure~\ref{fig:robotScenario} is $3$-transducer live, i.e., when the human corresponds to a transducer with three states, the robot can move so as to figure out which two (out of three) charging stations the human may block.
The significance of this new notion is that focusing on $k$-transducer liveness reduces the complexity of the synthesis problem to co-NP-completeness. 
With the new definition, we are condensing the problem of finding out if it is possible for the system to gradually adapt to the environment's behavior to a more manageable complexity class. While the complexity is still beyond polynomial time, we can employ efficient SAT solvers after encoding the specification to a synthesis game to determine if some simple environment behavior can block the synthesized system from satisfying the specification.

The synthesized controllers in our approach iterate through possible transducers for the environment behavior and use own behavior adapted to a particular environment transducer until the environment is found to react inconsistently with the supposed transducer. Whenever this is found to be the case, the controller switches to the next possible own behavior. 
Such a strategy can be implemented on a Turing machine that uses at most polynomial space (in the size of the environment's transducer). 
We leave the problem of computing controller implementations that adapt to the environment as quickly as possible for future work.

In the context of learning the behavior of a bounded environment, the synthesis for \emph{absolute liveness} properties -- properties that are insensitive to additions of prefixes, was shown to be contained in EXPTIME~\cite{boundedtemp}. There are games that are $k$-transducer live but not absolutely live.
For example, in Figure~\ref{fig:robotScenario}, the game is not absolutely live but is $3$-transducer live.  
Additionally, if the game is absolutely live and is winning for the system player against $k$-transducers, then the game is also $k$-transducer live. Thus, absolute liveness is a 
stronger assumption when compared to $k$-transducer liveness for the purpose of safely probing bounded environments. 

The paper is structured as follows. In Section~\ref{sec:prelims}, we provide preliminaries that define two-player games used for reactive synthesis, transducers and transducer languages. In Section~\ref{sec:lower}, we prove a PSPACE lower bound of synthesis for bounded environments by encoding a quantified Boolean formula in a B\"uchi game against a $k$-transducer environment. In Section~\ref{sec:live}, we introduce the notion of $k$-transducer liveness. We show that identifying whether a game is $k$-transducer live is co-NP-complete. 
The co-NP-containment proof is constructive and shows how to compute a strategy to win such games.
Lastly, we conclude in Section~\ref{sec:conclude}.



\section{Preliminaries}
\label{sec:prelims}

Let $\Sigma$ and $\Gamma$ be finite alphabets. Furthermore, let $\A =  \Sigma\Gamma$.

\begin{definition}[Parity Game]
\label{def:bgame}
A game $G$ between two players \player{1} and \player{2} is a tuple $\langle V,\Sigma,\Gamma,E,\iota,F \rangle$, where
\begin{itemize}
    \item $V = V_1 \uplus V_2$ is the set of vertices (or positions). $V_1$ is the set of \player{1} vertices and $V_2$ is set of \player{2} vertices.
    \item $\Sigma$ and $\Gamma$ are the action sets of \player{1} and \player{2}, respectively.
    \item $E: (\{V_1 \times \Sigma \} \cup \{V_2 \times \Gamma\}) \to V$ is the transition 
    function, where
    \begin{align*}
    &E(u,a) = E_1(u,a), \text{ if } u \in V_1 \text{ and } a \in \Sigma \text{ and } \\
    &E(v,b) = E_2(v,b), \text{ if } v \in V_2 \text{ and } b \in \Gamma.    
    \end{align*}
    Here, $E_1 : V_1 \times \Sigma \to V_2$ and $E_2: V_2 \times \Gamma \to V_1 $ are 
    functions corresponding to transitions from $P_1$ and $P_2$ vertices, respectively. 
    \item $\iota \in V_1$ is the initial vertex.
    \item $F : V \to \mathbb{N}$ is a coloring function.
\end{itemize}
\end{definition}
A \emph{play} $\rho$ is a (possibly infinite) sequence $u_0v_0u_1v_1\dots$ of vertices such that $u_0 = \iota$ and there is a sequence of actions $w = a_0b_0a_1b_1 \dots$ such that $E(u_i,a_i) = v_i$ and $E(v_i,b_i) = u_{i+1}$ ($i \in \mathbb{N})$. Moreover, we say that the play $\rho$ is \emph{generated} by the word $w$.  A play $\rho = u_0v_0u_1v_1\dots$ is winning for player~2 if and only if the largest number occurring infinitely often in the sequence $F(u_0)\,F(v_0)\, F(u_1) \ldots$ is even. A \emph{B\"uchi} game is a variant of the parity game such that $F : V \to \{1,2\}$.
In a B\"uchi game, a play $\rho$ is winning for $P_2$ if some vertex $v$ such that $F(v) = 2$ is repeated infinitely often.
Throughout this paper, all plays that are not winning for \player{2} are winning for \player{1}.
Lastly, \emph{reachability} games are a variant of Büchi games. For them, a play $\rho$ is winning for $P_2$ if it eventually reaches a vertex $v$ such that $F(v) = 2$.

A strategy $\sigma$ for a player $P \in \{P_1,P_2\}$ maps every finite prefix sequence of actions $w \in \A^* \cup \A^* \Sigma$ ending with a action for the respective other player to a next action of $P$ (where for \player{1}, $\sigma$ also maps the empty word to an element of $\Sigma$).
A word $w \in \A^* \cup \A^\omega$ is said to \emph{agree} with a strategy~$\sigma$ for \player{1} if the play $\rho$ generated by $w$ agrees with $\sigma$. A strategy $\sigma$ for player $P$ is said to be winning if every play $\rho$ (starting at $\iota$) that agrees with $\sigma$ is winning for player $P$.

We model finite-state reactive systems with inputs in $\Gamma$ and outputs in $\Sigma$ by \emph{transducers}. We use these transducers to model bounded memory environments.
\begin{definition}[Transducer]
\label{def:trans}
A finite transducer $T$ is a tuple $\langle \Sigma, \Gamma, M,s, L,\eta \rangle$ that consists of the following components:
\begin{itemize}
    \item $M$ is a finite set of states,
    \item $\Sigma$ is a set of alphabet symbols, called the output alphabet,
    \item $\Gamma$ is a set of alphabet symbols, called the input alphabet,
    \item $s \in M$ is the initial state,
    \item $\eta: M \times \Gamma \to M$ is a function, called the transition function and
    \item $L: M \to \Sigma$ is a function, called the labeling function.
\end{itemize}
\noindent We extend $\eta$ to words in $\Gamma^*$ in the straight-forward way. Thus, $\eta: \Gamma^* \to M$ is such that $\eta(\epsilon) = s$ and for $x \in \Gamma^*$  and $i \in \Gamma$, $\eta(x\cdot i) = \eta(\eta(x),i)$. We define the labeling function on words in $\Gamma^*$ , $\widehat{L}:~\Gamma^* \to \Sigma$, as $\widehat{L} = L \circ \eta$.
\end{definition}

Each transducer $T$ induces a \emph{strategy} $f_T: \Gamma^* \to \Sigma$, where for all $w = a_0 b_0 a_1 b_1 \ldots a_n b_n\in (\Sigma \times \Gamma)^*$, we have that $f_T(w) = \widehat{L}(b_0 b_1 \ldots b_n)$. Thus, $f_T(w)$ is the action that $T$ outputs after reading the \player{2} actions in $w$.
A transducer with $k$ states is called a $k$-\emph{transducer}. 
Furthermore, a strategy induced by a $k$-transducer is called a \emph{$k$-transducer strategy}.

In the subsequent sections, we analyze the behavior of \player{1} when it is restricted to using $k$-transducer strategies. For this purpose, we define compatibility of plays with respect to some $k$-transducer strategy for \player{1} as follows.

\begin{definition}[Agreement with a $k$-transducer]
\label{def:agree}
A word $w = a_0 b_0 a_1 b_1 \ldots \in \A^* \cup \A^* \Sigma \cup  A^\omega$ is said  to agree with a $k$-transducer $T=\langle \Sigma, \Gamma, M,s, L,\eta \rangle$ if for every prefix $a_0 b_0 \ldots a_n b_n a_{n+1}$ of $w$, we have that $\widehat{L}(b_0 b_1 \ldots b_{n})=a_{n+1}$.
\end{definition}
\noindent We define $\A^{*}_k$ ($\A^\omega_k$) to be the set of words $w \in \A^{*} (\A^\omega)$ that agree with some $k$-transducer for~\player{1}. 

\begin{definition}[$k$-transducer language]
\label{def:klang}
We define the $k$-transducer language for a reachability, B\"{u}chi, or parity game $G$, denoted by $\lang_k(G)$, to be the set of words 
$\A^\omega$ that agree with some $k$-transducer $T$ for \player{1} and for which the play generated by $w$ is winning for \player{2}.
\end{definition}

Games of infinite duration, as in Definition~\ref{def:bgame}, are a conceptual model to reduce \emph{reactive synthesis} to the task of finding out for a given specification in some temporal logic such as linear temporal logic (LTL) whether there exists a transducer whose executions all satisfy the given specification, without the possibility to control the input to the transducer~\cite{Bloem2018}. 
The specification is translated to an automaton over infinite words, which is in turn translated to a game of infinite duration such that the implementations that satisfy the specification are exactly the \emph{system player} strategies in such games. Games with this property are also called \emph{synthesis games} for the respective specification. We will use Büchi games for hardness proofs and parity games for complexity class containment proofs in this paper. The former are a special case of parity games, so that hardness results carry over to the parity game case. 

%

\begin{problemstatement}
Given a reachability, Büchi, or parity game $G = \langle V,\Sigma,\Gamma,E,\iota,F \rangle$, does the system (\player{2}) have a winning strategy in the game $G$ against an $k$-transducer environment (\player{1})?
\end{problemstatement}

We note that despite stating our results based on a formalization of the reactive synthesis problem using games in this paper, the environment model definition is the same as the one by Kupferman et al.~\cite{boundedtemp}. Therefore,  hardness and containment results in the size of the environment transducers are valid for the reactive synthesis in bounded environments problem as well.

\section{General B\"{u}chi games against bounded adversaries}
\label{sec:lower}

We address \emph{reactive synthesis for bounded environments} by studying the problem from Problem Statement~1, which reformulates this variant of reactive synthesis in the scope of games.
The current lower bound for the complexity of this problem is NP-hard~ \cite{boundedtemp}. We improve this lower bound by showing that the problem is at least PSPACE-hard in this section.

To show the PSPACE-hardness of this problem, we encode a quantified Boolean formula~(QBF) formula $\psi$ of the form 
$$\psi =  \forall x_1 \exists y_1 \forall x_2 \exists y_2 \dots \forall x_k \exists y_k:~(C_1 \land C_2 \land \dots \land C_r)$$ 
into a reachability game $G_\psi$ on a graph of size $O(k\cdot (r+1))$ that is winning for the system player (\player{2}) if and only if $\psi$ is valid (equivalent to $\mathbf{true}$). 
The environment player (\player{1}) can use only $(k+1)$-transducer strategies ($k \in \mathbb{N}$), while the system player has no such restrictions.
In the above QBF, $C_1, \ldots, C_r$ are the \emph{clauses}, i.e., disjunctions of literals in $\{x_i,\neg x_i, y_i, \neg y_i :1 \leq i \leq k\}$.
\begin{figure}[hptb]
\begin{tikzpicture}[->,>=stealth',shorten >=1pt,auto,node distance=2.5cm, semithick,scale = 0.6]
    
\tikzstyle{playeronenode}=[shape=rectangle,minimum width=0.8cm,minimum height=0.8cm,inner sep=-40pt,draw]
\tikzstyle{playertwonode}=[shape=diamond,minimum width=0.9cm,minimum height=0.9cm,inner sep=-40pt,draw]
    
\node[playeronenode,initial,fill=cyan] (start) {$ \iota$};
\node[playertwonode,right of = start] (enode) {$n$};
\path (start) edge node {$e$} (enode);
\node[inner sep = 0.2cm,fit = (start)(enode),draw] (E) {};
\node[above right,blue] at (E.north west) {Phase~1};

\node[playeronenode, below of = start,fill = cyan] (A1) {$a_1$};
\path (enode) edge node [anchor = south] {$*$} (A1); 
\node[playertwonode,right of = A1] (A2) {$a_2$};
\node[right of = A2] (A3) {\quad $\dots$};
\node[playeronenode,right of = A3] (A4) {$a_{2k-1}$};
\node[playertwonode,right of = A4] (A5) {$a_{2k}$}; 
\path (A1) edge node {$x_1,\lnot x_1$} (A2);
\path (A2) edge node {$y_1,\lnot y_1$} (A3);
\path (A4) edge node {$x_k, \lnot x_k$} (A5);

\node[inner sep = 0.2cm,fit = (A1)(A2)(A3)(A4)(A5),draw] (A) {};
\node[above right,blue] at (A.north west) {Phase~2};

\node[playeronenode, below of = A1,fill = cyan] (x11) {$v_{1,1}^\bot$};
\node[below of = x11] (x11d) {};
\node [playertwonode, right of = x11] (y11f) {$v_{2,1}^\bot$};
\node [playertwonode, below of = y11f] (y11t)    
{$v_{2,1}^\top$};
\node [playeronenode, right of = y11f] (x21f) {$v_{3,1}^\bot$};
\node [playeronenode, below of = x21f] (x21t) {$v_{3,1}^\top$};
\node [playertwonode, right of = x21f] (y21f) {$v_{4,1}^\bot$};
\node [playertwonode, below of = y21f] (y21t) {$v_{4,1}^\top$};


\path (A5.south) edge node [anchor = south] {$y_k,\lnot y_k$} (x11.north east);
\path (x11) edge node {$x_1$} (y11f);
\path (x11) edge node {$\lnot~x_1$} (y11t);
\path (y11f) edge node {$\lnot y_1$} (x21f);
\path (y11f) edge node {$y_1$} (x21t);
\path (y11t) edge node {$y_1, \lnot y_1$} (x21t);
\path (x21f) edge node {$x_2$} (y21f);
\path (x21f) edge node {$\lnot x_2$} (y21t);
\path (x21t) edge node {$x_2,\lnot x_2$} (y21t);

\node[inner sep = 0.2cm,fit = (x11)(y11f)(y11t)(x21f)(x21t)(y21f)(y21t),draw,loosely dashed] (C1) {};


\node[playeronenode, below of = y21t,fill = cyan] (x12) {$v_{1,2}^\bot$};
\node[below of = x12] (x12d) {};
\node [playertwonode, left of = x12] (y12f) {$v_{2,2}^\bot$};
\node [playeronenode, left of = y12f] (x22f) {$v_{3,2}^\bot$};
\node [playeronenode, below of = x22f] (x22t) {$v_{3,2}^\top$};
\node [playertwonode, left of = x22f] (y22f) {$v_{4,2}^\bot$};
\node [playertwonode, below of = y22f] (y22t) {$v_{4,2}^\top$};

\node[inner sep = 0.2cm,fit = (x12)(y12f)(x22f)(x22t)(y22f)(y22t),draw,loosely dashed] (C2) {};

\path (y21t) edge node [anchor = east] {$y_2,\lnot y_2$} (x12);

\path (x12) edge node {$x_1,\lnot x_1$} (y12f);
\path (y12f) edge node {$y_1$} (x22f);
\path (y12f) edge node {$\lnot y_1$} (x22t);
\path (x22t) edge node {$x_2,\lnot x_2$} (y22t);
\path (x22f) edge node {$x_2$} (y22t);
\path (x22f) edge node {$\lnot x_2$} (y22f);

\node[playertwonode, right of = x12, fill = orange] (P21) {};
\node[playeronenode, right of = P21, fill = orange] (P22) {};
\node[inner sep = 0.4cm,fit = (P21)(P22),draw,loosely dashed] (P2) {};
\node[above left] at (P2.north east) {$P_1$ paradise};
\path (P21) edge [bend left] node {$*$} (P22);
\path (P22) edge [bend left] node {$*$} (P21);

\node[playeronenode, below of = P21, fill = green] (P11) {};
\node[playertwonode, below of = P22, fill = green] (P12) {};
\node[inner sep = 0.4cm,fit = (P11)(P12),draw,loosely dashed] (P1) {};
\path (P11) edge [bend left] node {$*$} (P12);
\path (P12) edge [bend left] node {$*$} (P11);
\node[below left] at (P1.south east) {$P_2$ paradise};

\node[playeronenode, below of = y22t,fill = cyan] (x13) {$v_{1,3}^\bot$};
\node [playertwonode, right of = x13] (y13f) {$v_{2,3}^\bot$};
\node [playertwonode, below of = y13f] (y13t)    
{$v_{2,3}^\top$};
\node [playeronenode, right of = y13f] (x23f) {$v_{3,3}^\bot$};
\node [playeronenode, below of = x23f] (x23t) {$v_{3,3}^\top$};
\node [playertwonode, right of = x23f] (y23f) {$v_{4,3}^\bot$};
\node [playertwonode, below of = y23f] (y23t) {$v_{4,3}^\top$};

\node[inner sep = 0.2cm,fit = (x13)(y13f)(y13t)(x23f)(x23t)(y23f)(y23t),draw,loosely dashed] (C3) {};

\path (y22t) edge node [anchor = west] {$y_2,\lnot y_2$} (x13);
\path (x13) edge node {$x_1$} (y13t);
\path (x13) edge node {$\lnot x_1$} (y13f);
\path (y13t) edge node {$y_1,\lnot y_1$} (x23t);
\path (y13f) edge node {$y_1$} (x23f);
\path (y13f) edge node {$\lnot y_1$} (x23t);
\path (x23f) edge node {$x_2,\lnot x_2$} (y23f);
\path (x23t) edge node {$x_2,\lnot x_2$} (y23t);

\path (y23t) edge node [anchor = west] {$y_2, \lnot y_2$} (P11);
\path (y23f) edge node [anchor = west] {$y_2$} (P11);

\path (C3.north) edge [red,dashed,thick] node {$e$} (P21);
\path (C2) edge [red,thick,dashed] node {$e$} (P21);
\path (C1.south east) edge [red,dashed,thick] node {$e$} (P21);

\node[inner sep = 0.2cm,fit = (C1)(C2)(C3),draw] (phase3) {};
\node[above right,blue] at (phase3.north west) {Phase~3};

\node[inner sep = 0.09cm,fit = (P1)(P2),draw] (phase4) {};
\node[right,blue] at (phase4.east) {\rotatebox{90}{Phase~4}};
\end{tikzpicture}

\caption{The Reachability game corresponding to the quantified Boolean formula (QBF) $\psi = \forall x_1 \exists y_1 \forall x_2 \exists y_2:(\neg x_1 \vee y_1 \vee \neg x_2) \wedge (\neg y_1 \vee x_2) \wedge (x_1 \vee \neg y_1 \vee y_2)$.  
\player{1} plays from squared vertices and \player{2} plays from diamond-shaped vertices. 
The \emph{paradises} for the two players are positions from which they (corresponding player) is guaranteed to win a play. 
The dashed edges labeled by the exit action $e$ correspond to edges that start at every \player{1} controlled vertex in phase~3 to the \player{1} paradise, except that there is no such transition from $v_{1,1}^\bot$. Lastly, if for some action of a player, an outgoing edge is not shown, we assume that this action leads to the \emph{paradise} of its opponent.
}
\label{fig:winf}
\end{figure}
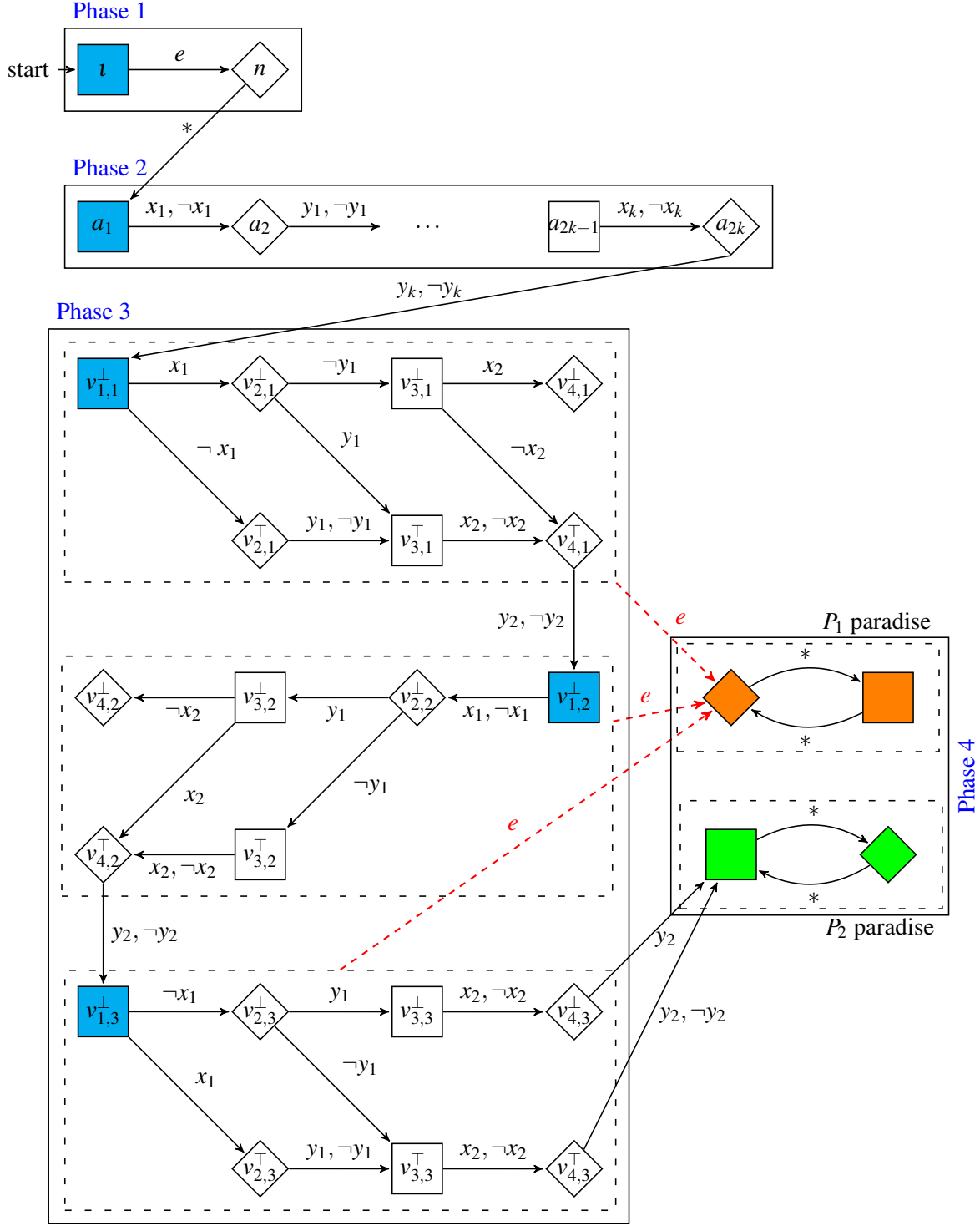
We depict the structure of the game for a specific QBF instance in Figure~\ref{fig:winf}.

In the game, \player{1} chooses its actions from the set $\{x_i, \neg x_i : 1 \leq i \leq k\} \cup \{e\}$, while \player{2} uses $\{y_i, \neg y_i : 1 \leq i \leq k\}$ as action set.
Intuitively, \player{1} and \player{2} make assignments to the $x$-variables and $y$-variables, respectively, using their
$\neg x_i/x_i$ and $\neg y_i/y_i$ actions.
\player{1} is trying to satisfy the formula, while \player{2} is trying to falsify the formula. Additionally, \player{1} has the possibility of playing an \emph{exit} move ($e$), which ends the assignment process.
The game is played in four phases. 
\begin{enumerate}
    \item In the first phase, \player{1} demonstrates that it can play the $e$ action.
    \item In the second phase, \player{1} and \player{2} jointly construct an assignment.
    \item In the third phase, the satisfaction of each of the clauses with respect to the assignment is checked.
    \item In the final phase, the play ends in a \emph{paradise} for one of the players.
\end{enumerate}
Additionally, if in the first two phases, one of the players plays an illegal action (from any state moves that are not shown in the figure), then the play moves to a paradise (for the opponent) immediately.

The most interesting phase of the game is the third one, which is played in $r$ stages, with  stage $1 \leq j \leq r$ corresponding to clause $C_{j}$. In this phase, the players have to make the assignments for all their variables turn by turn, once for each clause. Thus, for every clause they each have to play $k$ rounds. 

For every $1 \leq j \leq r$ and $1 \leq i \leq 2k$ such that $i$ is odd, the vertex $v_{i,j}^\top$  represents the fact that \player{1} is currently choosing the value of variable $x_{(i+1)/2}$ in the clause $C_j$ and that the variable values chosen so far already satisfy the clause. Similarly, the vertex $v_{i,j}^\bot$ represents the fact that \player{1} is currently choosing the value of variable $x_{(i+1)/2}$ in the clause $C_j$ and the variable values chosen so far do not already satisfy the clause. From even $i$, \player{2} is choosing the value for variable $y_{i/2}$.
A special case is the $e$ action of player~\player{1}, which immediately leads to a paradise for one of the two players. If \player{1} plays it during phase~2, then the play reaches a paradise for \player{2} and hence, \player{2} wins the game. If \player{1} plays $e$ during phase~3, then \player{1} wins the game (by reaching a paradise).

The transducer for \player{1} has $k+1$ states. 
In every transducer that wins the game (for \player{1}), the initial state $s = m_0$ is reserved for generating the $e$ action. The other $k$ states are needed to make an assignment to the $x$-variables during the second phase of the game. Without loss of generality, we denote the transducer state used for assigning to variable $x_i$ by $m_i$.
This is because the actions required by \player{1} to make an assignment to the variable $x_i$ are unique for each $i$, i.e., the other actions (including $e$) are illegal. Therefore, all the $k+1$ states of the transducer have to be used by \player{1} for the game to enter phase~3. Lastly, the only legal actions for \player{1} from $v_{1,1}^{\bot}$ are $x_1$ and $\lnot x_1$. In particular, the exit action $e$ is not allowed. Since \player{1} is forced to use unique memory states to make assignments for the $x$-variables in phase~1 and the number of memory states ($k+1$) available for it is exactly equal to the number of different actions required to complete the first two phases, at the start of phase~3, the transducer for \player{1} hits the loop $m_1 \rightarrow m_2 \dots \rightarrow m_k \rightarrow m_1$ on the (memory) states. If it does not hit this loop on its memory states, then it (\player{1}) would have made an illegal action causing it to lose.
However, in the vertices in phase~3, \player{1} wins the game if it plays the $e$ action.

At this point, if \player{2} continues playing the same assignment (that it had generated in phase~2), the loop on the environment's (\player{1}) transducer will be continued. However, any deviation from the initial assignment by \player{2} could trigger a return to the transducer state $m_0$ that outputs $e$, thereby making \player{2} lose the game. Hence, \player{2} is forced to make the same assignments. 

\begin{theorem}
The QBF formula $\psi$ is valid if and only if \player{2} has a winning strategy against any $(k+1)$-transducer for \player{2} in the reachability game constructed from $\psi$. 
\label{lemma:winf_pspace}
\end{theorem}
\begin{proof}
\noindent ($\Rightarrow$) 
If the formula $\psi$ is valid, then for each $1 \leq i \leq k$, values for $y_i$ can be chosen based only on the values of $x_0, y_0, \ldots, x_{i-1}$ such that after assigning values to all Boolean variables in $\psi$, all clauses are satisfied.
We observe that this gives rise to a strategy for \player{2} to win the game built from $\psi$. When~\player{2} plays this strategy, at the end of the clause $C_j$ component of the game in phase~3, vertex $v_{1,j+1}^\bot$ is reached. 
For \player{1}, the only way to then avoid eventually reaching the \player{2} paradise is by playing an $e$ move. Since the transducer only has $k+1$ states, after $k+1$ decisions, the transducer has to visit an old state again, in particular the state $m_1$ giving the value for $x_1$. Thus, if \player{2} repeats the same actions as before, the transducer (for the environment) is then forced to repeat its actions as well, preventing it from ever using the exit action $e$.
 
\noindent ($\Leftarrow$) 
 For the other direction, we show that if the QBF formula is not valid, then for each strategy of \player{2}, there is a counter-strategy for \player{1} that lets \player{1} win. The strategy for \player{1} is defined as follows:
 \begin{itemize}
 \item In phase 1, \player{1} initially plays $e$ from the transducer state $m_0$.
 \item In phase 2, \player{1} plays the values of $x_i$ from state $m_i$ (for $1 \leq i \leq k$) suitable to eventually falsify some QBF clause for the choices of $x_1,y_1, \ldots, x_{i-1}, y_{i-1}$ made by \player{1} and \player{2} for  $x_1, \ldots, y_{i-1}$ thus far. 
By the assumption that the QBF is not valid, such a choice is guaranteed to exist for every $1 \leq i \leq k$.
 \item Finally, in phase 3, \player{1} repeats the values of  $x_1, \ldots, x_k$ indefinitely unless  \player{2} chooses values other than $y_1, \ldots, y_k$, in which case it plays the exit move $e$ by transitioning to $m_0$. 
 \end{itemize}
 Note that here, the choice of the strategy for \player{1} depends on the strategy for \player{2}. Since we are only asking if there exists a strategy for \player{2}, this is a valid line of reasoning. By following the above \player{1} strategy, the play  reaches a \player{1} paradise, either due to some clause being violated or due to \player{1} using the exit action action  $e$ in the third phase (since that indicates that the values of $y_1, \ldots, y_k$ have changed).
 \end{proof}

The PSPACE lower bound for the synthesis under bounded environments problem shows that from a computational complexity point of view, the problem is somewhat difficult. 
In the next section, we provide a method to reduce the complexity by strengthening the synthesis problem. We do so by requiring that the system is able to gather information about the environment's transducer in \emph{safe ways}.


\section{$k$-transducer liveness}
\label{sec:live}

In the literature, a linear-time temporal (LTL) property $\psi$  is said to be \emph{live} if and only if for all partial computations $\alpha$, there is a (possibly infinite) sequence of states $\beta$ such that $\alpha \beta \models \psi$, i.e., no partial execution is irremediable: it always remains possible for the required `good thing' to occur in the future~\cite{alpern1985oct}. 

When employing synthesis games that encode the reactive synthesis problem and restricting the environment to behavior implementable as $k$-transducers, we can rephrase liveness for such environment transducers on the level of games as follows:

\begin{definition}[$k$-transducer liveness]
\label{def:klive}
A game $G$ is $k$-transducer live if $$\forall \alpha \in \A^*_k :  \exists \beta \in \A^{\omega}:~ \alpha \beta \in \mathcal{L}_k(G).$$
\end{definition}
For safely observing the behavior of the environment, we require that no finite play that is generated by a $k$-transducer environment can cause a system failure irrespective of the choice of moves made by the system so far. 
The above definition of liveness ensures that any finite play agreeing with some $k$-transducer strategy can always be extended to a play that is winning for the system~(\player{2}).
For example, the scenario in Figure~\ref{fig:robotScenario} is $3$-transducer live.

We call a reactive synthesis problem instance \emph{$k$-transducer live} if the synthesis games encoding the problem instance are $k$-transducer live. Since $k$-transducer liveness is given on the level of action sequences, either all synthesis games for a specification have this property, or none of them have.
Next, we analyze the complexity of determining if a given game is $k$-transducer live.

\begin{theorem}
\label{thm:coNPContainment}
Deciding whether a parity game $G$ is $k$-transducer live is contained in co-NP.
\end{theorem}
\begin{proof}
We show that it suffices to check for all possible $k$-transducers separately that \player{2} can win from every reachable combination of game position and transducer state.
Hence, a co-NP algorithm can non-deterministically guess a transducer $T$ and perform this check. 
The game is $k$-transducer live if and only if
for all transducers, the answer is ``yes''. Containment in co-NP follows from this observation.

For each non-deterministically guessed $k$-transducer $T=\langle \Sigma, \Gamma, M,s, L,\eta \rangle$ and a given game $G = \langle V,\Sigma,\Gamma,E,\iota,F \rangle$, the co-NP algorithm builds a new parity game $\widetilde G = \langle \widetilde V,\Sigma,\Gamma,\widetilde E,\widetilde \iota,\widetilde F \rangle$ in which \player{1}'s moves are forced to be compatible with the transducer $T$ 
and the game is played on a graph with vertex set $\widetilde V = V \times M$.

Without loss of generality, let $\top \in V$ be a vertex of \player{2} in $G$ from which \player{2} cannot lose a suffix play. If such a vertex does not exist in the graph, we just add it to $G$. 
For every $m \in M$, $u \in V_1$, $v \in V_2$, $a \in \Sigma$, and $b \in \Gamma,$ we define:
\begin{align*}
&\widetilde E((u,m),a)  = \begin{cases} (E_1(u,a),m), & \text{ if } a=L(m), \\
(\top,m) & \text{ otherwise.} \end{cases} \\
&\widetilde E((v,m),b)  = (E_2(v,b),\eta(m,b)). \\
&\widetilde \iota  = (\iota,s). \\
&\widetilde F((u,m))  = F(u). \\
&\widetilde F((v,m))  = F(v).
\end{align*}

In the game $\widetilde G$, \player{1} is restricted to play exactly the strategy induced by $T$. If \player{1} does not follow this strategy, then it loses. 
It can be tested if \player{2} has a strategy against $T$ from every position $(v,m)$.
Since \player{1}'s action from  every transducer state is fixed, the game becomes a deterministic parity automaton, for which the emptiness of the languages of the automaton's states can be determined in time polynomial in the size of the game \cite{DBLP:conf/fossacs/KingKV01}.
To complete the proof, we now prove the following two sub-claims:
\begin{description}
\item[\textbf{Claim:~1}]
\label{claim1}
If for some transducer $T$, there exists a position $(v,m)$ reachable from $\tilde \iota$ in $\widetilde G$ that is losing for \player{2}, then $G$ is not $k$-transducer live.
\item[\textbf{Claim:~2}]
\label{claim2}
If $G$ is not $k$-transducer live, then there exists a transducer $T$ such that some position $(v,m)$, reachable from $\widetilde \iota$, is losing for \player{2} in $\widetilde G$.
\end{description}

\noindent\textbf{Proof of Claim~1:}
Let $T$ be the transducer, and $(v,m)$ be a position (in $\widetilde G$) reachable under the sequence $\alpha$ from which \player{2} loses.
We construct an extension $\alpha_2$ of $\alpha$ such that all transducers $T'$ that agree with $\alpha_2$ behave identical to $T$. 
Since $(v,m)$ is losing for \player{2} in the game $\widetilde G$, this means that $\alpha_2 \beta \notin \mathcal{L}_k(G)$ for every possible choice of $\beta$.

Let $\mathcal{T}'$ be the set of $k$-transducers that are compatible with $\alpha$. Note that this set is finite as $k$ is constant. 
Either all transducers $T' \in \mathcal{T}'$ behave identically to $T$ after reading $\alpha$ (in which case we are done), or there is at least one transducer $T' = (M',\Sigma,\Gamma,s',\eta',L')$ that does not.
If it does not, then there is a finite word $b = b_0 \ldots, b_n \in \Gamma^*$ of length at most $k^2$ such that feeding $b$ to both $T$ and $T'$ from the respective states reached after $\alpha$ forces a different a output symbol of the two transducers $T$ and $T'$. 
We now extend $\alpha$ to $\alpha' = \alpha \, \widehat{L}(\alpha|_\Gamma) b_0 \widehat{L}(\alpha|_\Gamma b_0) b_1 \ldots \widehat{L}(\alpha|_\Gamma b_0 \ldots b_{n-1}) b_n$ (recall that $\widehat{L}$ is the labeling function of the transducer extended to words). 

The position $(v',m')$ reached (in $\widetilde G$) under $\alpha'$ is still losing for \player{2}. This is because in the game $\widetilde G$, \player{1}'s actions are fixed. Thus, all the positions reachable from positions that are themselves losing for \player{2} still remain losing. However, when considering $\alpha'$ instead of $\alpha$, the set $\mathcal{T}'$ does not contain $T'$ any more. Note that since we only extended $\alpha$, no new elements can be added to $\mathcal{T}'$ in this way.
Since $\mathcal{T}'$ is finite, $\alpha'$ can be extended in this way until only transducers that behave identically to $T$ are compatible with $\alpha'$. The claim follows. \qed

\noindent\textbf{Proof of Claim~2:}
If $G$ is not $k$-transducer live, then there exists a prefix word $\alpha$ that is compatible with a $k$-transducer $T$ such that no suffix word $\beta$ exists so that $\alpha \beta$ induces a winning play for \player{2} and $\alpha \beta$ can be generated by some transducer. In particular, for all suffixes $\beta \in A^\omega$ such that $\alpha\beta$ is compatible with $T$, the corresponding play is losing for \player{2}. 

Let $\widetilde G$ be the game generated from $T$, and $(v,m)$ be the position reached in $\widetilde G$ reached under $\alpha$. Thus for all $\alpha\beta \in A^\omega$ compatible with $T$, the play $\beta$ starting from $(v,m)$ is losing for \player{2}. Moreover, any winning word ($\beta \in A^\omega$) for \player{2}  starting from $(v,m)$ in $\widetilde G$ results in a word $\alpha\beta \in \mathcal{L}_k(G)$, which contradicts the assumption on $\alpha$. This means that \player{2} loses the parity game from $(v,m)$.
\end{proof}

Adapting the above theorem to address the reachability case is simple, as the winning condition is only used for the fact that non-emptiness checking of a deterministic automaton is not harder than polynomial time. 
As a consequence of the above theorem, we can verify whether a given game is $k$-transducer live rather quickly. The following theorem shows that \player{2} (the system) can always win games that are $k$-transducer live. 
Additionally, we also show that such a strategy (represented as a Turing machine) requires space only polynomial in $k$ and the number of positions in the game. 
Overall, this shows that in $k$-transducer live games, the system can learn the behavior of the environment and adapt its own behavior without violating the objective encoded into the game.

\begin{theorem}
\label{thm:rep_pspace}
Any $k$- transducer live game $G$ is always winning for \player{2}. Furthermore, there exists a winning strategy for \player{2} that requires at most polynomial space ($Poly(n,k)$), where $n$ is the number of vertices in $G$. Moreover, if $G$ is a reachability game, by following this strategy, \player{2} can reach a final state in at most an exponential number of steps $(Poly(n),~EXP(k))$.
\end{theorem}
\begin{proof}
We arrange the set of all possible transducers into a sequence $T_1, \ldots, T_N$ ($N \leq EXP(k)$) in lexicographic order.

Our strategy iterates over the transducers $T_i = (M_i,\Sigma,\Gamma,s_i,\eta_i,L_i)$ while generating a single sequence of decisions $b_0 b_1 \ldots \in \Gamma^\omega$.
For every transducer $T_i$, the strategy checks if for the current vertex $v \in V$ in the game $G$ played, the game $\widetilde G$ built according to the construction in the proof of Theorem~\ref{thm:coNPContainment} has position $(v,m)$ reachable for some $m \in M_i$. If that is not the case, this means that $T_i$ cannot be the environment strategy, and the system player strategy moves to the next transducer $T_{i+1}$ in the sequence.
Otherwise, the strategy starts to maintain a set $M'$ of transducer states that the transducer can currently be in. Initially, these are all states $m \in M_i$ for which $(v,m)$ is reachable in $\widetilde G$.

Afterwards, the strategy picks one particular state $m \in M'$ as the current conjectured environment's transducer state. It computes a winning strategy from $(v,m)$ in $\widetilde G$ in polynomial space and time, which is lasso-shaped. 
The strategy chooses the actions from this lasso while the actions of \player{1} agree with the lasso, while also simultaneously updating the candidate set $M'$ of current states.
Once the actions of \player{1} do not agree with the \player{1} actions along the lasso any more, it is known that either the environment does not play $T_i$, or the current state of $T_i$ is not $m$. State $m$ is then removed from $M'$.
If at some point $M'$ becomes empty, the strategy moves to the next transducer $T_{i+1}$.

If a game is $k$-transducer live, the system always wins - eventually, the correct transducer $T_i$ is found with the correct current state $m$, or \player{2} manages to play a winning lasso already earlier. The winning strategy for \player{2} against $(T_i,m)$ ensures that \player{2} wins when eventually, the $(T_i,m)$ combination is considered. By $k$-transducer liveness, the prefix play until then does not prevent \player{2} from winning once the correct $(T_i,m)$ pair has been found.
A corresponding strategy for reachability games works in the same way. For each $(T_i,m)$, any deviating behavior of the environment is detected in at most $O(n \cdot k)$ many steps, as this is the maximum lasso length of the \player{2} strategy computed in $\widetilde G$. 

The number of transducer/state pairs is exponential in $k$. Hence, for reachability games, a final state is visited after a number of steps at most exponential in $k$ and linear in~$n$.
\end{proof}

In Theorem~\ref{thm:coNPContainment}, we showed the upper complexity bound for detecting $k$-transducer liveness. The following theorem establishes co-NP-completeness for determining whether a B\"uchi game is $k$-transducer live.
\begin{theorem}
\label{thm:klive_hardness}
Deciding whether a B\"{u}chi game $G$  is $k$-transducer live is co-NP-hard.
\end{theorem}
\begin{proof}
Let $\psi = C_1 \land C_2 \land \dots \land C_r$ be a Boolean formula in conjunctive normal form (CNF) over the set $X = \{x_1,x_2,\dots,x_k\}$ of variables. We construct a reachability game $G_\psi$ corresponding to the formula $\psi$ such that $G_\psi$ is $k$-transducer live if and only if $\psi$ is not satisfiable. The underlying game graph $G_\psi$ is shown in Figure~\ref{fig:co_np}. 
The action set for \player{1} is $\{\top_1,\bot_1,\top_2,\bot_2,\dots,\top_k,\bot_k\}$. For $b \in \{\top,\bot\}$, \player{1} uses the action $b_i$ to assign the value $b$ to the variable $x_i$. 
The action set for \player{2} is $\{\epsilon\}$, corresponding to a dummy move. Note that in this way, \player{2} has no role to play. 

\begin{figure}[p]
    \centering
    \begin{tikzpicture}[->,>=stealth',shorten >=1pt,auto,node distance=2cm, semithick,scale = 0.6]
    \tikzstyle{playeronenode}=[shape=rectangle,minimum width=0.8cm,minimum height=0.8cm,inner sep=-40pt,draw]
    \tikzstyle{playertwonode}=[shape=diamond,minimum width=0.9cm,minimum height=0.9cm,inner sep=-40pt,draw]
    
    \node[state, initial,playeronenode,fill = cyan] (C1S) {$\iota$};
    \node[below of = C1S] (C1S1) {};
    \node[below of = C1S1] (C1S2) {};
    \node[state, right of = C1S,playertwonode] (C1X1T) {$x_{1,1}^\bot$};
    \node[state, below of = C1X1T, anchor = south,playertwonode] (C1X1F) {$x_{1,1}^\top$};
    \node[state, playeronenode, right of = C1X1T] (C1X2T) {$\set{y}{1,1}{\bot}$};
    \node[state, playeronenode, below of = C1X2T, anchor = south] (C1X2F) {$\set{y}{1,1}{\top}$};
    \node[right of  = C1X2T] (C1E1) {$\dots$};
    \node[right of  = C1X2F, anchor = south] (C1E2) {$\dots$};
    
    
    \node[state,right of = C1E2, anchor = south,playeronenode] (TC1XMF) {$\set{y}{k,1}{\bot}$};
    \node[inner sep = 0pt, fit = (TC1XMF),draw,white] (C1T) {};
    
    \node[inner sep = 0.3cm,fit = (C1S)(C1T)(C1X2F),draw,loosely dashed] (C1) {};
    
    \path (C1S) edge node {} (C1X1T);
    \path  (C1S) edge node [anchor = east] {} (C1X1F);
    \path (C1X1T) edge node {$\epsilon$} (C1X2T); 
    \path  (C1X1F) edge node [anchor = north] {$\epsilon$} (C1X2F); 
   
    \node[below left] at (C1.north west) {Clause~$C_1$}; 
    
    \node[state,playeronenode, below of = C1S1,fill=cyan] (C2S) {$y_{k,1}^\top$};
    \node[below of = C2S] (C2S1) {};
    \node[below of = C2S1] (C2S2) {};
    \node[state, right of = C2S,playertwonode] (C2X1T) {$\set{x}{1,2}{\bot}$};
    \node[state, below of = C2X1T, anchor = south,playertwonode] (C2X1F) {$\set{x}{1,2}{\top}$};
    \node[state, right of = C2X1T,playeronenode] (C2X2T) {$\set{y}{1,2}{\bot}$};
    \node[state, below of = C2X2T, anchor = south,playeronenode] (C2X2F) {$\set{y}{1,2}{\top}$};
    \node[right of  = C2X2T] (C2E1) {$\dots$};
    \node[right of  = C2X2F, anchor = south] (C2E2) {$\dots$};
    \node[below of = C2E2] (C2E3) {$\vdots$};
  
    \node[state, right of = C2E2, anchor = south, playeronenode] (TC2XMF) {$\set{y}{k,2}{\bot}$};
    \node[inner sep = 0.35cm, fit = (TC2XMF),draw,white] (C2T) {};

    \node[inner sep = 0.3cm,fit = (C2S)(TC2XMF) (C2X2F),draw,loosely dashed] (C2) {};
    
    \path (C2S) edge node [anchor = north] {} (C2X1T);
    \path  (C2S) edge node [anchor = east] {} (C2X1F);
    \path (C2X1T) edge node {$\epsilon$} (C2X2T); 
    \path  (C2X1F) edge node [anchor = north] {$\epsilon$} (C2X2F);

    \node[below left] at (C2.north west) {Clause $C_2$};

    \node[playertwonode, right of = TC2XMF, fill = green] (P21) {};
    \node[playeronenode, right of = P21, fill = green] (P22) {};
    \node[above right,yshift=12pt,xshift=-25pt] at (P22.north east) {$P_2$ paradise};
    \path (P21) edge [bend left] node {$\epsilon$} (P22);
    \path (P22) edge [bend left] node {$*$} (P21);

    
    \node[state,playeronenode, below of = C2S2,fill=cyan] (CMS) {$y_{k,m}^{\top}$};

    \node[state, right of = CMS,playertwonode] (CMX1T) {$\set{x}{1,m}{\bot}$};
    \node[state, below of = CMX1T, anchor = south,playertwonode] (CMX1F) {$\set{x}{1,m}{\top}$};
    \node[state, right of = CMX1T,playeronenode] (CMX2T) {$\set{y}{2,m}{\bot}$};
    \node[state, below of = CMX2T, anchor = south,playeronenode] (CMX2F) {$\set{y}{2,m}{\top}$};
    \node[right of  = CMX2T] (CME1) {$\dots$};
    \node[right of  = CMX2F, anchor = south] (CME2) {$\dots$};
   
    \node[state, right of = CME1,playeronenode]  (TCMXMT) {$\set{y}{k,m}{\bot}$};
    \node[state, below of = TCMXMT, anchor = south,playeronenode] (TCMXMF) {$\set{y}{k,m}{\top}$};
    
    \node[inner sep = 0.3cm,fit = (CMS)(TCMXMT)(TCMXMF),draw,loosely dashed] (CM) {};
    
    
    \path (CMS) edge node [anchor = north] {} (CMX1T);
    \path  (CMS) edge node [anchor = east] {} (CMX1F);
    \path (CMX1T) edge node {$\epsilon$} (CMX2T); 
    \path  (CMX1F) edge node [anchor = north] {$\epsilon$} (CMX2F);

    \node[below left] at (CM.north west) {Clause $C_r$};
    
    \node[playeronenode, right of = TCMXMF, fill = orange] (P11) {};
    \node[playertwonode, right of = P11, fill = orange] (P12) {};
    \path (P11) edge [bend left] node {$\epsilon$} (P12);
    \path (P12) edge [bend left] node {$*$} (P11);
    \node[below right,yshift=-12pt,xshift=-25pt] at (P12.south) {$P_1$ paradise};
    
    \path (TC1XMF) edge node {$*$} (P21);
    \path (TC2XMF) edge node {$*$} (P21);
  
    \path (TCMXMT) edge node {$*$} (P21);
    
    \path (TCMXMF) edge node {$*$} (P11);
    
    \end{tikzpicture}
    \caption{The game graph for the reachability game $G_\psi$ corresponding to $\psi$. The initial vertex $\iota$ is marked as such. \player{1} plays from squared vertices and \player{2} plays from diamond-shaped vertices. The objective of \player{1} is to reach the orange vertex in the \player{1} paradise, i.e., satisfy the formula $\psi$. Dually, the objective of \player{2} is to reach a green vertex in the \player{2} paradise, i.e., falsify the formula $\psi$ (however, there is no role for \player{2} in the game, i.e., the moves of \player{2} have no effect on the satisfaction of the clauses). Whether the current clause is already satisfied is represented in the superscript of the vertex label.
    From vertices with label $y_{i,j}^{\top}$ or $y_{i,j}^{\bot}$, \player{1} makes the assignment for the next variable $x_{i+1}$ (for clause~$j$), using one of the actions $\top_{i+1}$ and $\bot_{i+1}$. Suppose \player{1} makes the move $a \in \{\top_{i+1},\bot_{i+1}\}$ from $y_i^{\alpha}$ for clause $C_j$, then the next state is $x_{i+1,j}^{\beta}$, where $\beta = \top$, if setting $x_{i+1}$ to $b$ makes clause $C_j$ satisfied, and $\beta = \alpha$ otherwise. 
   }
    \label{fig:co_np}
\end{figure}
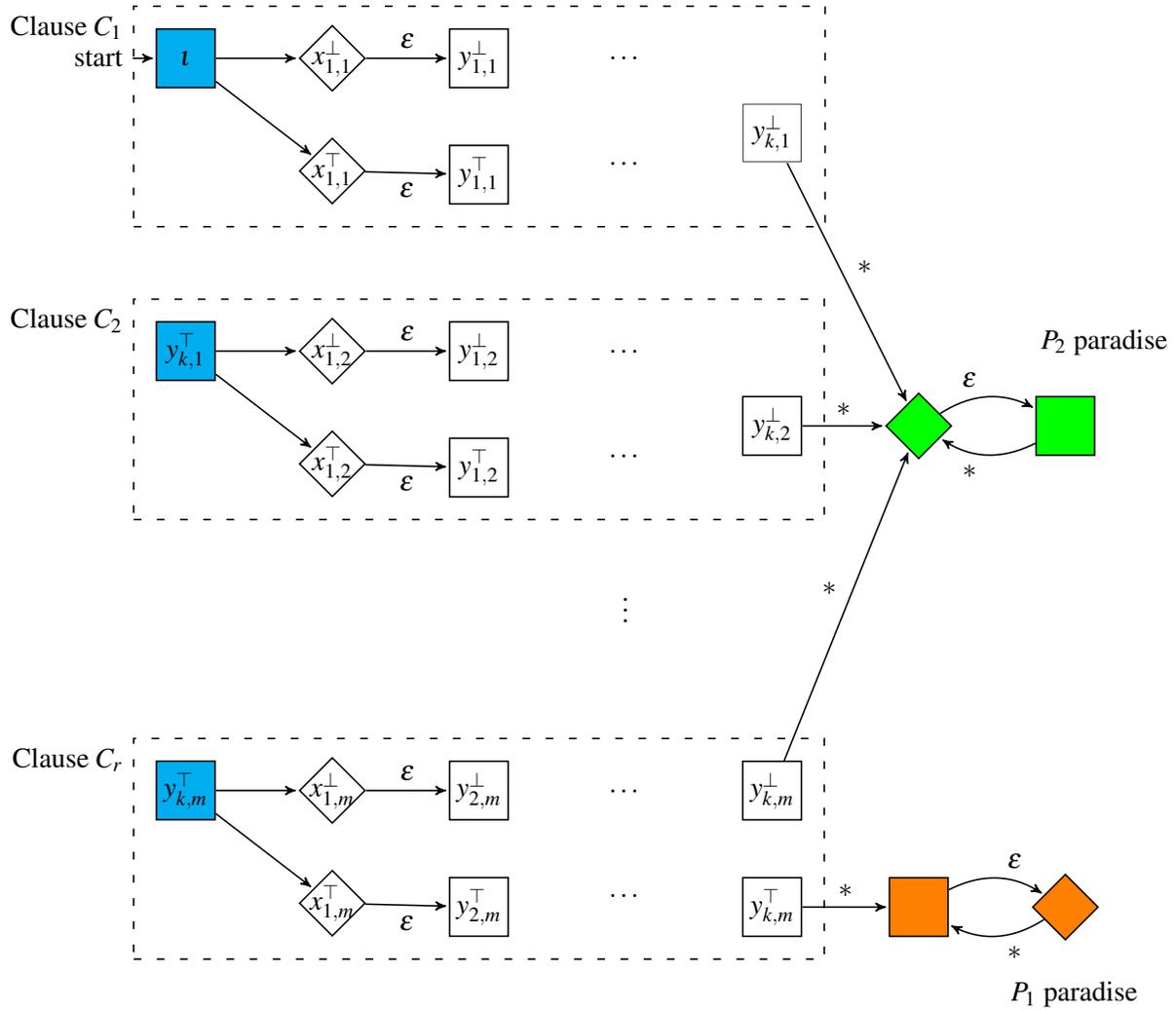
We now prove that the formula $\psi$ is satisfiable if and only if the game on $G_\psi$ is $k$-transducer live.
\player{1} is restricted to $k$-transducer strategies. Since there are $k$ variables, \player{1} needs exactly one transducer state for each of the variables. In the game, \player{1} makes an assignment for each variable, once per clause. \player{2} repeats the same dummy ($\epsilon$) move for each of its vertex. Say \player{1} is currently assigning the variable $x_{i}$ the value $b$ for clause $j$, then the evaluation of this clause, denoted by $eval_j(i)$, is computed as $eval_j(i) = eval_j(i-1) \lor \llbracket C_j\rrbracket_{x_{i} = b}$. Here, $\llbracket C_j \rrbracket_{x_{i} = b}$ is 
$\top$ if setting $x_i$ to $b$ makes clause $C_j$ satisfied, else it is $\bot$.

If clause $j$ is satisfied, then the play moves on to clause $C_{j+1}$ and \player{1} starts assigning values to the set~$X$ once again. The fact that \player{1} can only use $k$-transducer strategies implies that \player{1} cannot change the assignments.
If some clause is not satisfied, then the play moves to a green vertex and \player{2} wins. On the other hand, if all the clauses are satisfied, then \player{1} wins, i.e., the play reaches the orange vertex and can never reach a green vertex. 

Thus, if $\psi$ is satisfiable by a particular assignment for $X$, then $P_1$ uses this assignment to obtain the corresponding $k$-transducer to stay safe in $G_\psi$. Thus, $G_\psi$ is not $p$-transducer live. Otherwise, if $\psi$ is not satisfiable, there does not exist any $k$-transducer $T$ for $P_1$ such that the play generated by the strategy $f_T$ does not eventually hit a green state. This is because the assignment corresponding to $T$ cannot satisfy some clause $C_j$ and the game hits the green vertex in the first clause where this happens. 
This implies that $G_\psi$ is $k$-transducer live.
\end{proof}
The above proof of co-NP-hardness is similar to the proof of NP-completeness of decision problems for partial-observation games with mean-payoff objectives from \cite{Chatterjee2020} (Lemma~4). In~\cite{Chatterjee2020}, \player{1} makes the assignment to the variables and \player{2} chooses the clauses to check the assignment. We adopted this proof idea for the $k$-transducer liveness problem.
The main difference is that in our reduction, we remove the role of \player{2} and require \player{1} to repeatedly make the same assignments for each clause. The fact that \player{1} is restricted to $k$-transducers ensures that it makes the same assignments for each clause.


\section{Conclusion}
\label{sec:conclude}

In this paper, we studied the problem of synthesizing controllers that satisfy given specifications when used in environments of a bounded size. While this problem was originally introduced by Kupferman Et al.~\cite{boundedtemp}, our work was motivated by applications in robotics. This problem has a simple definition and yet captures the idea that a high-level robot controller should operate correctly in environments with unknown dynamics of bounded complexity. The problem is also interesting on its own because it captures the idea that a controller may observe the environment's behavior to adapt to it.

We provided two results for this synthesis problem at the level of games, as reactive synthesis is commonly reduced to solving games. 
The game formulation enables us to understand a controller to be synthesized and the environment's behavior as strategies of the two players in the game, thereby simplifying the exposition. 
Our first result is negative: we strengthened the NP lower bound given by Kupferman et al.~\cite{boundedtemp} to PSPACE. This PSPACE lower bound means that we cannot hope to employ a satisfiablity (SAT) solver
for this problem. Such solvers have proven their applicability for a plethora of practical scenarios in the last two decades, so being able to use them would have helped to scale reactive synthesis under bounded environments to scenarios of practical interest.

We identified the system player's necessity to strategize to find out how it should probe the environment's behavior ``safely'' as the key reason for this high complexity. To address this issue, we defined the notion of $k$-transducer liveness, which captures games and reactive synthesis problem instances in which such strategic reasoning is not needed. Consequently, the system player only has to care about satisfying the objective encoded into the game once the environment transducer is safely found. 
We proved that this modified problem is co-NP-complete, thereby showing that its complexity is comparably lower.
As an added benefit, our co-NP solving algorithm that is given in the proof of Theorem~\ref{thm:coNPContainment} can be implemented with a satisfiability (SAT) solver. After guessing an environment transducer, it performs an analysis of a graph that is the product of the guessed transducer and the given game. This graph a one-player game, which is conceptually the same as a deterministic automaton. Building such a product in a SAT instance is already done in exact SAT-based minimization of deterministic automata \cite{DBLP:conf/lpar/BaarirD15,DBLP:conf/sat/Ehlers10}, except that in our case, reachability of product game positions also needs to be considered. The positive results obtained for deterministic automaton minimization in the past suggest a reasonable scalability of SAT-based $k$-transducer liveness game solving and $k$-transducer liveness reactive synthesis. 

Our work focused on identifying computational complexities to prepare high-level robotics applications in unknown environments. As such, we had to exclude some practical considerations from this work, which we leave for future work.
Our approach for $k$-transducer live games
computes strategies that
run through all the possible environment transducers. This makes it slow to converge to the final behavior (for a fixed environment strategy). Finding ways of improving this approach will make solving $k$-transducer live games more useful for practical applications.
Although our choice to cast  the complexity of the behavior of the environment transducer as its number of states is a natural one (in computer science), making our analysis of the problem  interesting from a theoretical viewpoint, this choice can be further  honed from a practical perspective.
One way of doing this would be to synthesize a transducer that works correctly if the environment \emph{eventually} follows a fixed transducer (where the system is not able to observe when this happens). This idea has already been applied for synthesizing implementations that are robust against deviations from environment assumptions~\cite{DBLP:conf/hybrid/EhlersT14,DBLP:conf/nfm/Ehlers11}. 
It was shown that due to the fact that the synthesis algorithms always compute finite-state implementations, these implementations have to start working towards the satisfaction of the specification even before the environment can be observed to have stabilized. 
A similar effect can be expected for synthesizing implementations that perform early best-effort specification satisfaction in environments of bounded complexity as well. Hence, analyzing the problem of reactive synthesis for environments that eventually follow a bounded transducer appears to be worthwhile.

\nocite{Neider2011Oct}
\nocite{bounedsynthesis}
\nocite{ryana_lassos}

\bibliographystyle{eptcs}

\bibliography{ref}

\end{document}